\documentclass[11pt,letterpaper]{amsart}
\usepackage{amsmath, amssymb, amsfonts, amsthm, graphicx, hyperref, color}
\usepackage{graphics, verbatim, subfig, caption, footnote, lscape}

\newtheorem{thm}{Theorem}[section]

\newtheorem{lem}[thm]{Lemma}
\newtheorem{prop}[thm]{Proposition}
\theoremstyle{remark}
\newtheorem{rem}[thm]{Remark}
\numberwithin{equation}{section}

\newcommand{\set}[1]{\left\{#1\right\}}
\newcommand{\Real}{\mathbb R}
\newcommand{\Natural}{\mathbb N}

\newcommand{\prob}{\mathbb{P}}

\newcommand{\ud}{\mathrm d}

\newcommand{\qll}{\textsf{qll}\,}
\newcommand{\pqll}{\textsf{pqll}\,}

\title{Maximum Penalized Quasi-Likelihood Estimation of the Diffusion Function}

\author{Jeff Hamrick}
\address[Jeff Hamrick]{Department of Mathematics and Computer Science, Rhodes College, 2000 N. Parkway, Memphis, TN 38112, USA}
\email{hamrickj@rhodes.edu}

\author{Yifei Huang}
\address[Yifei Huang]{Department of Mathematics and Statistics, Boston University, 111 Cummington Street, Boston, MA 02215, USA}
\email{yifei@bu.edu}

\author{Constantinos Kardaras}
\address[Constantinos Kardaras]{Department of Mathematics and Statistics, Boston University, 111 Cummington Street, Boston, MA 02215, USA}
\email{kardaras@bu.edu}

\author{Murad S. Taqqu}
\address[Murad S. Taqqu]{Department of Mathematics and Statistics, Boston University, 111 Cummington Street, Boston, MA 02215, USA}
\email{murad@bu.edu}

\thanks{The third and fourth author gratefully acknowledge partial support by the National Science Foundation, grant numbers DMS-0908461 and DMS-1007616, respectively.}%

\date{\today}%

\begin{document}

\begin{abstract}
We develop a \emph{maximum penalized quasi-likelihood estimator} for estimating in a nonparametric way the diffusion function of a diffusion process, as an alternative to more traditional kernel-based estimators. After developing a numerical scheme for computing the maximizer of the penalized maximum quasi-likelihood function, we study the asymptotic properties of our estimator by way of simulation. Under the assumption that overnight London Interbank Offered Rates (LIBOR); the USD/EUR, USD/GBP, JPY/USD, and EUR/USD nominal exchange rates; and 1-month, 3-month, and 30-year Treasury bond yields are generated by diffusion processes, we use our numerical scheme to estimate the diffusion function.
\end{abstract}

\setcounter{section}{-1}

\maketitle

\section{Introduction}
\label{s:Introduction}

One of the key achievements in the field of financial engineering is the representation of the price of contingent claims as the expectation of their discounted future payoffs under the so-called \emph{risk-neutral probability}, i.e., a probability measure under which discounted (by the exponential of integrated short rate) traded asset prices are martingales. In complete markets, this representation permits the computation of the unique arbitrage-free price, as well as the hedging strategy needed to remove all the risk associated with issuing (or writing) the contingent claim.

From a practical viewpoint, there is a major issue with choosing the form of the risk-neutral probability measure, as it depends on the model specification. In other words, the stochastic movement of asset prices must be modeled in a way that is consistent with observed data collected from the market. As equivalent changes of probability leave the quadratic variation of a process intact, market data provide a possible way to pin down the volatility component; then, the drift rate under the risk-neutral measure is simply set equal to the applicable short rate.

The most elementary continuous-time model for asset prices is probably geometric Brownian motion, for which the log-price dynamics are characterized by both constant drift rate and constant volatility. This particular choice of parameters is not consistent with various asset prices like stock prices, exchange rates, or interest rates. Therefore, more elaborate models have to be utilized. One direct generalization of geometric Brownian motion are local volatility models, where the diffusion coefficient is a function of the underlying asset's level.

In this paper, we consider the class of diffusion models with coefficient functions that are allowed to depend on the asset level. We aim to develop a nonparametric estimation procedure for the diffusion function based on a maximum penalized quasi-likelihood method, following the work of \cite{eggermont:lariccia:2001}, \cite{eggermont:lariccia:2009}.

In Section \ref{s:MarketModel}, we consider a one-dimensional diffusion model for the movement in the price of a financial asset, and briefly review the history of attempts to estimate the diffusion function $\sigma$. In Section \ref{s:VolatilityEstimationForDiffusionProcess}, we develop a quasi-likelihood function for a diffusion process and then add a penalization term (also known as a \emph{regularization} term) to obtain a \emph{penalized quasi-likelihood function}. After establishing the existence of a maximizer $\theta_{*}$ of the penalized quasi-likelihood function, we use techniques from the calculus of variations to justify a property of $\theta_{*}$ which, in turn, permits us to introduce a numerical scheme for calculating $\theta_{*}$ along discrete, non-uniform design points. In Section \ref{s:SimulationStudy}, we study two simulated diffusion processes and note that the mean integrated squared error of our estimator converges to zero at rates that seem to be comparable to the rates of convergence achieved by kernel-based estimators. Then, in Section \ref{s:ApplicationToRealData}, we use software 
to estimate the diffusion function for overnight London Interbank Offered Rates (LIBOR), 1-month and 30-year Treasury bond yields, and the USD/EUR, USD/GBP, and JPY/USD exchange rates.

\section{Discussion of the Problem}
\label{s:MarketModel}

\subsection{The market model}
We consider a diffusion model for the price movement of a financial asset. In particular, we study a one-dimensional diffusion $(Y_{t})_{t \in [0, T]}$ with $Y_{0}=y$ and dynamics given by
\begin{equation}
\label{eq:Diffusion}
\ud Y_{t}=b(Y_{t}) \ud t+\sigma(Y_{t}) \ud W_{t}, \quad t \in[0,T].
\end{equation}
Here, $W$ is a standard Brownian motion, $b: \mathbb{R}\mapsto \mathbb{R}$ and $\sigma: \mathbb{R} \mapsto \mathbb{R}_{++} \equiv (0, \infty)$ are Borel-measurable functions, and $T>0$ is a fixed time horizon. With $C([0,T]; \Real)$ denoting the canonical path-space of continuous functions equipped with the Borel sigma field, the dynamics described in (\ref{eq:Diffusion}) are valid under the probability $\mathbb{P}_{y}^{(b,\sigma)}$ on $C([0,T]; \Real)$, which is such that $\mathbb{P}_{y}^{(b,\sigma)} [Y_{0}=y]=1$.

In order for the problem to be well-posed, we assume that the stochastic differential equation (\ref{eq:Diffusion}) has a weak solution that is unique in the sense of probability law. For conditions that guarantee the existence of a weak solution to (\ref{eq:Diffusion}), see \cite[Theorem 5.4, pg. 332]{karatzas:shreve:1991}. Of course, uniqueness of a weak solution is a necessary condition for the well-posedness of statistical estimation problems involving diffusions.

\begin{rem}
The assumption that the functions $b$ and $\sigma$ have domain $\Real$ is made because in the analysis we shall use a conditional Gaussian approximation for the transition densities of the diffusion $Y$. In practice, the diffusion can live on any sub-interval of $\Real$, such as $(0, \infty)$; this situation will create only theoretical obstacles. Ultimately, it is usually the case that a scale transformation of $Y$ will result in a related diffusion with full support on $\Real$; for example, a log-transformation of a diffusion supported on $(0, \infty)$ results in a new diffusion supported on $\Real$.
\end{rem}

\subsection{Nonparametric estimation of the diffusion function}

We consider the problem of estimating in a nonparametric way the coefficient $\sigma$ that appears in equation (\ref{eq:Diffusion}) under the assumption that $\mathbb{P}_{y}^{(b,\sigma)}$ is the historical (or statistical, or real-world) probability law. The estimation of $\sigma$ is of crucial importance, since it stays fixed under the equivalent changes of probability measure that are necessary for pricing and hedging contingent claims. More specifically, recall that when we pass from the historical probability law to the risk-neutral probability law, $\sigma$ remains unaltered while $b$ changes.

In the case in which we \emph{continuously} observe data over the time interval $[0,T]$, i.e., when the whole path $(Y_{t})_{t\in[0,T]}$ is observed, perfect estimation of $\sigma$ is possible, at least in the window of observations $[m_{T},M_{T}]$, where we define
\begin{equation}
\label{eq:MinimumAndMaximum}
m_{T}:=\min_{t\in[0,T]}Y_{t} \quad \textrm{and} \quad M_T:=\max_{t\in[0,T]}Y_{t}.
\end{equation}
Indeed, the quadratic variation  of $Y$ under $\prob^{(b,\sigma)}$ is $\langle Y,Y \rangle = \int_{0}^{\cdot}\sigma^{2}(Y_{t}) \ud t$ and hence $\sigma^{2}(Y_{t})=\partial \langle Y,Y \rangle_{t}/\partial t$ for $t \in [0, T]$. In this case of continuous observations, given perfect estimation of $\sigma$, \cite{kutoyants:2004} contains substantial results regarding estimation of the drift $b$ in both parametric and nonparametric settings.

Of course, the case of continuous observations is only a theoretical idealization. In practice, we are usually presented with discrete observations $Y_{t_{0}},Y_{t_{1}}, \ldots, Y_{t_{n}}$, where $0=t_{0}< \ldots<t_{n}=T$. Typically, researchers study two distinct problems for the case in which we only have discrete observations at our disposal. In the case of \emph{low sampling frequency}, the sampling interval remains fixed, and the behavior of estimators is studied as $T$ tends to infinity. (In \cite{bandi:phillips:2003}, this is referred to as the \emph{long-time asymptotic} approach --- see, for example, \cite{gobet:hoffmann:reib:2004} for the nonparametric approach in this case.) We shall consider the problem of estimating the diffusion function $\sigma$ when we have \emph{high sampling frequency}, that is, the mesh of the partition, as defined by $\max_{i=1, \ldots,n}|t_{i}-t_{i-1}|$, tends to zero but the time horizon $T$ remains fixed. See \cite{bandi:phillips:2003} for more on nonparametric estimation of the drift and the diffusion function in the case of high-frequency data, which is referred to as the \emph{infill asymptotic} approach. The infill asymptotic approach has recently been combined with observation of the integrated diffusion process to estimate the drift and diffusion functions --- see \cite{comte:genoncatalot:rozenholc:2007}.

The earliest research on the problem of nonparametric estimation of the diffusion function $\sigma: \mathbb{R}\mapsto (0,\infty)$ seems to have been undertaken by \cite{florens-zmirou:1993}. In \cite{jacod:2000}, kernel-smoothing estimation of $\sigma$ is utilized in a manner motivated by kernel estimation of the density functions, and a rate of convergence of the estimator of order $n^{-m/(2m+1)}$ is obtained for the case in which $\sigma$ is $m$ times continuously differentiable. A similar approach is undertaken in \cite{banon:1978}, \cite{jiang:knight:1997} and \cite{soulier:1998}. In \cite{hoffmann:1999a}, an $L_{p}$-loss estimator for $\sigma$ is developed which, under suitable conditions, has a minimax rate of convergence that is also of order $n^{-m/(2m+1)}$. The same author constructs estimators for the coefficients of a diffusion process based on adaptive wavelet thresholding (with respect to an unknown degree of smoothness in the coefficient functions) in \cite{hoffmann:1999b}. Wavelet methods were also used in \cite{genoncatalot:laredo:picard:1992} to estimate the drift and diffusion functions and again, the rate of convergence was found to match the classical rate of convergence for kernel-based estimators. A kernel smoothing approach is used in \cite{fan:fan:lv:2007} to define spatial and temporal estimators for $\sigma$. This technique was later adapted by \cite{hamrick:taqqu:2009} to test diffusions for stationarity.

In this paper, we propose instead to use a method based on the maximum penalized quasi-likelihood. We provide estimators of $\sigma$ that empirically behave at least as well as the other estimators developed in the literature up to the present time. We note additionally the connection between the method that we will propose and kernel estimators via reproducing kernel Hilbert spaces, as is presented, for example, in \cite[pp. 20-26]{eggermont:lariccia:2009}.

\section{Estimation Using Maximum Penalized Quasi-Likelihood}
\label{s:VolatilityEstimationForDiffusionProcess}

\subsection{Zero drift} The whole analysis below will be carried out assuming that $b \equiv 0$. Analysis for other values for the drift coefficient could be similarly carried out, albeit in a more complicated way. However, this would not eventually serve any purpose in the case of high-frequency data. In fact, any analysis regarding consistency and rates of convergence that applies to the case of zero drift applies also immediately to the case of non-zero drift. The reason is that, in view of Girsanov's theorem, under very mild integrability assumptions on $b / \sigma$, the probabilities $\mathbb{P}^{(b, \sigma)}$ and $\mathbb{P}^{(0, \sigma)}$ are equivalent. In fact, in Section \ref{s:SimulationStudy} we shall present a simulation study that explores the efficiency of our proposed diffusion coefficient estimator in the case of a diffusion with non-zero drift.

\subsection{Quasi-likelihood} Under the assumption that $b\equiv0$ and that our observations are of high frequency, for each $i=1, \ldots,n$, the conditional law of $Y_{t_{i}}$ given $(Y_{t_{0}},Y_{t_{1}}, \ldots,Y_{t_{i-1}})$ is approximately\footnote{Obviously, the quality of the approximation will improve as the mesh tends to zero.} normal with mean $Y_{t_{i-1}}$ and standard deviation $\sigma(Y_{t_{i-1}})\sqrt{t_{i}-t_{i-1}}$. Use of this approximation permits us to construct the \emph{weighted quasi-log-likelihood} of the sample, which is defined by
\begin{equation}
\label{eq:QuasiLogLikelihood}
\qll (\sigma \ | \ Y_{t_{0}}, \ldots,Y_{t_{n}}) = \frac{1}{n} \sum_{i=1}^{n} \Big\{- \log \Big( \sigma(Y_{t_{i-1}}) \Big) - \frac{1}{2} \Big(\frac{Y_{t_{i}}-Y_{t_{i-1}}}{\sigma(Y_{t_{i-1}})\sqrt{t_{i}-t_{i-1}}}\Big)^{2} \Big\}.
\end{equation}
(The symbol ``$\qll$'' for the above function stands for ``quasi-log-likelihood.'')

Of course, any estimator $\widehat{\sigma}$ that satisfies $\widehat\sigma (Y_{t_{i-1}})=|R_{t_{i-1}}|$ at the values in the observation set $\mathfrak{O} := \{Y_{t_{0}},Y_{t_{1}}, \ldots,Y_{t_{n - 1}} \}$, where
\begin{equation}
\label{eq:PointwiseVariances}
R_{t_{i-1}} := \frac{Y_{t_{i}}-Y_{t_{i-1}}}{\sqrt{t_{i}-t_{i-1}}}, \quad i =1, \ldots, n,
\end{equation}
would correspond to a maximum likelihood estimator with perfect fit to the data. (In particular, such an estimator is not unique.) As is typical in infinite-dimensional estimation problems (another example of which is nonparametric density estimation), naive interpolations of the points $\{(Y_{t_{i-1}},|R_{t_{i-1}}|)\}_{i = 1, \ldots, n}$ result in estimators that oscillate wildly and are nonsensical.
In order to effectively resolve this issue, one needs to impose some condition on the estimators. Nonparametric estimation procedures frequently call for the function being estimated to possess some degree of differentiability, which restores the well-posedness of the optimization problem. In this paper, we undertake such an approach, penalizing lack of smoothness in estimates of $\sigma$ via a \emph{maximum penalized quasi-likelihood} method.


\subsection{Maximum penalized quasi-likelihood estimator.} Since $\sigma>0$, the transformation
\[
\theta:=- \log(\sigma)
\]
is well-defined. The weighted quasi-likelihood is now
\begin{equation} \label{eq: qll}
\qll(\theta \ | \ Y_{t_{0}}, \ldots,Y_{t_{n}})=\frac{1}{n} \sum_{i=1}^{n} \Big\{ \theta(Y_{t_{i-1}}) - \frac{1}{2} R_{t_{i-1}}^{2}{e^{{2\theta (Y_{t_{i-1}})}}} \Big\}.
\end{equation}
For the numerical schemes that we will be developing, it is useful to rewrite the above expression in terms of the order statistics of the points in $\mathfrak{O}$. For $j=1, \ldots, n$, let $p_{j}$ denote the name of the rank $j$ point in $\mathfrak{O}$, with the smallest number having rank one. Set\footnote{We shall be assuming that there are no ties in the ranking. For an example of how we handle ties in the rankings, see Section \ref{s:ApplicationToRealData}. Note, however, that in theory ties occur with probability zero.}
\begin{equation}
\label{eq:OrderedData}
y_{j}:=Y_{p_{j}}\quad \textrm{and} \quad r_{j}:=R_{p_{j}}.
\end{equation}
For future reference, we set $y_0 := - \infty$ and $y_{n+1} := + \infty$; note that these two points do \emph{not} belong in our observations set $\mathfrak{O}$. With this new notation, we rewrite the weighted quasi-likelihood as
\[
\qll(\theta)=\frac{1}{n} \sum_{j=1}^{n} \Big\{ \theta(y_{j}) - \frac{1}{2} r_{j}^{2}{e^{{2\theta (y_{j})}}} \Big\},
\]
where we drop the dependence of $\qll$ on the sample. As mentioned before, we incorporate a term that will penalize estimates for lack of smoothness to produce the penalized quasi-log-likelihood
\begin{equation}
\label{eq:PenalizedQuasiLogLikelihood}
\pqll(\theta;m,\lambda)=\qll (\theta)-\frac{\lambda}{2}\int_{\mathbb{R}}|\theta^{(m)}(z)|^{2} \ud z,
\end{equation}
where $\lambda>0$ is a penalization factor, $m\in\mathbb{N}$, and $\theta^{(m)}$ is the derivative of order $m$ of $\theta$. We shall call $\theta_{*}$ a maximum penalized quasi-likelihood estimator if it is a solution to the problem
\[
\theta_{*}:=\arg \max_{\theta\in \mathbb{W}^{(m,2)}} \pqll(\theta;m,\lambda),
\]
where the maximization is over the space
\[
\mathbb{W}^{(m,2)}:=\Big\{ f:\mathbb{R}\longrightarrow \mathbb{R} \ \Big| \ \textnormal{$f^{(m-1)}$ exists, is absolutely continuous, and} \ \int_{\mathbb{R}}|f^{(m)}(z)|^{2} \ud z < \infty \Big\}.
\]

\noindent (Note that $\mathbb{W}^{(m,2)}$ as defined above is not the usual Sobolev space, which would be a particular subset of $\mathbb{L}^2 (\Real)$ --- functions in $\mathbb{W}^{(m,2)}$ might fail to be square-integrable.)

It can be shown that the maximizer $\theta_{*}$ of functionals like $\mathbb{W}^{(m,2)} \ni \theta \mapsto \pqll (\theta;m,\lambda)$ is a \emph{natural spline of order $2m-1$} with knots $\mathfrak{O} = \{y_1, \ldots, y_n \}$ --- see \cite{eubank:1988}.\footnote{The solution $\theta_{*}$ is an example of an $M$-type estimator --- see, for example, \cite{cox:1983}. For more information about the general problem of fitting splines with restricted sets of values using penalty functions, see \cite{utreras:1981}.} It follows that $\theta_{*}$ is at least $2m-1$ times differentiable, as well as piecewise polynomial of order $2m-1$ on all intervals $\{(y_{k}, y_{k+1})\}_{k=0, \ldots, n}$. (Recall that we are using the conventions $y_0 = - \infty$ and $y_{n+1} = + \infty$.) In particular, $\theta_*^{(2m-1)}$ is piecewise constant on the intervals $\{(y_{k}, y_{k+1})\}_{k=0, \ldots, n}$. The fact that $\theta_*$ is a natural spline means that all derivatives of order $m, \ldots, 2m-1$ vanish outside $[y_{1}, y_n]$. Since $\theta_*$ is at least $2m-2$ times continuously differentiable, this implies that $\theta_*^{(i)}(y_1) = 0 = \theta_*^{(i)}(y_n)$ for $i = m, \ldots, 2m-2$; furthermore, $\theta_*^{(2m-1)}(y_1 - ) = 0 = \theta_*^{(2m-1)}(y_n +)$, where $\theta_*^{(2m-1)}(y- )$ and $\theta_*^{(2m-1)}(y+ )$ will denote the left-hand and right-hand limit, respectively, of $\theta_*^{(2m-1)}$ at $y \in \Real$.

\begin{rem}
By definition, the estimators $\theta_{*}$ have a number of derivatives equal to zero at the points $y_1$ and $y_n$. Of course, the ``true'' $\theta = - \log (\sigma)$ is not expected to have such behavior near the endpoints. The method of penalization over-smooths the estimator close to the extreme observations $y_1$ and $y_n$ --- this situation is unavoidable, since there are very few observations near these points. It would be an interesting topic for future research to investigate in a rigorous way the severity of this effect close to the extreme observations.
\end{rem}

In all that follows, we shall be freely using the fact that $\theta_*$ has the above special structure. As $\theta_*^{(2m-1)}$ is not uniquely defined at the points $\mathfrak{O} = \{ y_1, \ldots, y_n\}$, we agree to pick $\theta_*$ such that $\theta_*^{(2m-1)}$ is right-continuous, i.e., we enforce $\theta_*^{(2m-1)} (y_i) = \theta_*^{(2m-1)} (y_i +)$ to hold for all $i = 1, \ldots, n$.

\begin{lem}
\label{thm:ConditionsOnTheta}

Let $\theta_{*}$ be the maximizer of equation (\ref{eq:PenalizedQuasiLogLikelihood}). Then, for any $\delta\in \mathbb{W}^{(m,2)}$,
\begin{equation}
\label{eq:FundamentalRelationship}
\frac{1}{n}\sum_{j=1}^{n}\Big\{ \delta(y_{j})\Big(1-r_{j}^{2}e^{2\theta_{*}(y_{j})}\Big) \Big\}=(-1)^{m-1}\lambda\int_{\mathbb{R}}\theta_{*}^{(2m-1)}(z)\delta'(z)d z.
\end{equation}

\end{lem}

\begin{proof}
To see this fact, assume that $\theta_{*}$ is the maximizer of equation (\ref{eq:PenalizedQuasiLogLikelihood}). For any $\delta\in \mathbb{W}^{(m,2)}$,
\[
0= \frac{\partial}{\partial\epsilon}\pqll(\theta_{*}+\epsilon \delta;m,\lambda)\Big|_{\epsilon=0}=\frac{\partial}{\partial\epsilon}\Big(\qll(\theta_{*}+\epsilon \delta)-\frac{\lambda}{2}\int_{\mathbb{R}}|\theta_{*}^{(m)}(z)+\epsilon \delta^{(m)}(z)|^{2}dz \Big)\Big|_{\epsilon=0}.
\]
By \eqref{eq: qll}, we obtain
\[
\frac{\partial}{\partial\epsilon} \qll(\theta_{*}+\epsilon \delta) \Big|_{\epsilon=0} = \frac{1}{n} \sum_{j=1}^{n} \Big\{ \delta(y_{j}) - \delta(y_{j})r_{j}^{2}{e^{{2\theta_{*}(y_{j})}}} \Big\}.
\]
Furthermore, it is straightforward that
\[
\frac{\partial}{\partial\epsilon}\Big(\frac{\lambda}{2}\int_{\mathbb{R}}|\theta_{*}^{(m)}(z)+\epsilon \delta^{(m)}(z)|^{2}dz \Big)\Big|_{\epsilon=0} = \lambda\int_{\mathbb{R}}\theta_{*}^{(m)}(z)\delta^{(m)}(z)dz.
\]
Therefore, the first-order conditions for optimality become
\[
\frac{1}{n} \sum_{j=1}^{n} \Big\{ \delta(y_{j}) - \delta(y_{j})r_{j}^{2}{e^{{2\theta_{*}(y_{j})}}} \Big\} = \lambda\int_{\mathbb{R}}\theta_{*}^{(m)}(z)\delta^{(m)}(z)dz.
\]
Since $\theta_{*}^{(m)}$ vanishes outside $[y_1, y_n]$, integration-by-parts implies that
\[
\lambda\int_{\mathbb{R}}\theta_{*}^{(m)}(z)\delta^{(m)}(z)dz = -\lambda \int_{\mathbb{R}}\theta_{*}^{(m+1)}(z)\delta^{(m-1)}(z)dz.
\]
Now, repeatedly integrating by parts and using the fact the $\theta_{*}^{(i)}$ vanishes outside $[y_1, y_n]$ for $i=m, \ldots, 2m-1$ yields the result.
\end{proof}

By considering appropriate functions $\delta$ and substituting them into equation (\ref{eq:FundamentalRelationship}), we obtain the following result, which will be the basis for our algorithm to compute $\theta_*$.

\begin{prop}
\label{thm:ExtentionToDeltaMass}
The maximum penalized quasi-likelihood estimator $\theta_*$ is such that
\begin{equation} \label{eq:ImplementationEquation}
\theta_{*}^{(2m-1)}(y)=\frac{(-1)^{m}}{n \lambda}\Big(k-\sum_{j=1}^{k}r_{j}^{2}{e^{{2\theta_{*}(y_{j})}}}\Big)
\end{equation}
holds for all $y\in[y_{k},y_{k+1})$, where $k=1, \ldots,n$. (Recall that $y_{n+1} = + \infty$, by convention.)
\end{prop}

\begin{proof}
Fix $y\in(y_{k},y_{k+1})$ for some $k \in \set{1, \ldots,n}$ and consider the function $\delta_{y} = \mathbb{I}_{(-\infty,y)}$. In that case, $- \delta_y'$ is a Dirac mass at $y$. Of course, equation (\ref{eq:FundamentalRelationship}) cannot be applied directly to a non-differentiable function like $\delta_{y}$. To circumvent this issue, we use an approximating procedure. Consider the sequence of functions $\Real \ni z \mapsto \delta_{y,N}(z) = \Phi(N (y - z))$ for $N \in \Natural$, where $\Phi$ is the cumulative distribution function of the standard normal distribution. Note that $\delta_{y,N} \in\mathbb{W}^{(m,2)}$ for $m\geq1$ and that $(\delta_{y,N})_{N \in \mathbb{N}}$ converges pointwise to $\delta_{y}$ at all points except for $y$. By equation (\ref{eq:FundamentalRelationship}) and typical arguments utilizing the distributional convergence of $\Real \ni z \mapsto - (\partial / \partial z ) \delta_{y, N}(z)$ to the Dirac mass at $y$, the equation
\[
\sum_{j=1}^{k}\Big\{ \delta_{y,N}(y_{j})\Big(1-r_{j}^{2}e^{2\theta_{*}(y_{j})}\Big) \Big\} = (-1)^{m-1}n\lambda\int_{\mathbb{R}}\theta_{*}^{(2m-1)}(z) \Big\{ \frac{\partial}{\partial z} \delta_{y,N}(z) \Big\} d z
\]
is valid for all $N \in \mathbb{N}$. It also leads to  equation (\ref{eq:ImplementationEquation}) upon sending $N$ to infinity as soon as one notices that $\lim_{N \to \infty} \sum_{j=1}^{k}  \delta_{y,N}(y_{j}) = k$ for $y \in (y_k, y_{k+1})$ and that $\theta_{*}^{(2m-1)}$ is constant on $(y_k, y_{k+1})$. Recalling that we are considering the right-continuous version of $\theta_*^{(2m-1)}$, we conclude that equation \eqref{eq:ImplementationEquation} is true for all $y \in [y_{k},y_{k+1})$.
\end{proof}

%

\subsection{Computing the estimator.} \label{ss: computing}
Assume that values for the penalization parameters $\lambda$ and $m$ are prespecified.  We will discuss possible ways of choosing $\lambda$ and $m$ in subsection \ref{ss:ChoosingThePenalizationParameters}.
We know $Y_{t_{0}},Y_{t_{1}}, \ldots,Y_{t_{n}}$ and hence the pairs $(y_{1},r_{1}),(y_{2},r_{2}), \ldots,(y_{n},r_{n})$ given by equation (\ref{eq:OrderedData}). We want to determine $\theta_*^{(i)}(y_{j})$ for $i=0, \ldots,2m-1$ and $j=1, \ldots, n$. These values of $\theta_*$ and its higher-order derivatives along the observation set $\mathfrak{O} = \{y_1, \ldots, y_n \}$ can be used to fit a spline, which then generates an estimate of $\theta$ (and, therefore, of $\sigma =\exp(-\theta)$ as well) on the interval $[y_1,y_n]$. The basic idea of the algorithm is to use equation (\ref{eq:ImplementationEquation}) to obtain $\theta^{(2m-1)}$ at a particular $y_{j}$ and then to ``work downwards'' to the lower-order derivatives.

We introduce further notation and definitions to support for the description of the iterative procedure induced by equation (\ref{eq:ImplementationEquation}). For $a = (a^{(0)}, \ldots, a^{(m-1)})$, define $\Theta(\cdot;a)$ to be the spline of order $2m-1$ with knots $\mathfrak{O}$, with the properies that $\Theta^{(i)}(y_{1};a)=a^{(i)}$ for $i=0, \ldots, m-1$, $\Theta^{(i)}$ vanishes on $(- \infty, y_1)$ for $i=m, \ldots, 2m-1$, and
\begin{equation}
\label{eq:IterativeEquation}
\Theta^{(2m-1)}(y;a)=\frac{(-1)^{m}}{n\lambda}\Big(k-\sum_{j=1}^{k}r_{j}^{2}{e^{{2\Theta(y_{j};a)}}}\Big)
\end{equation}
for all $y\in [y_{k},y_{k+1}),\ k=1, \ldots,n$, where again we are considering the right-continuous version of $\Theta^{(2m-1)}(\cdot;a)$ and set $y_{n+1} = \infty$ by convention. Note that the above properties characterize the spline $\Theta (\cdot; a)$ entirely in terms of a recursive procedure; we shall discuss how to compute all the values of $\Theta (\cdot; a)$ given $a = (a^{(0)}, \ldots, a^{(m-1)}) \in \Real^m$ in the sequel. For the time being, and as a warm-up for the algorithm that will be presented below, note that equation \eqref{eq:IterativeEquation} implies that
\begin{equation}
\label{eq:IterativeEquationBetter}
\Theta^{(2m-1)}(y_{k};a) = \Theta^{(2m-1)}(y_{k-1};a) + \frac{(-1)^{m}}{n\lambda} \Big(1 - r_{k}^{2} e^{{2\Theta(y_{k};a)}}\Big), \quad k=1, \ldots, n,
\end{equation}
upon agreeing that $\Theta^{(2m-1)}(y_{0};a) = 0$ as a matter of convention.

In view of Proposition \ref{thm:ExtentionToDeltaMass}, and since the estimator $\theta_*$ must be a natural spline, it will hold that $\theta_* (\cdot) = \Theta(\cdot; a_{*})$ where $a_{*} \in \Real^m$ is such that $\Theta^{(i)}(y_{n}; a_{*})=0$ for $i=m, \ldots,2m-1$. Therefore, our goal is to obtain the root of the nonlinear equation $F(a)=0$, where the mapping $F:\mathbb{R}^{m}\longrightarrow\mathbb{R}^{m}$ is defined by
\[
F(a):= \left(\Theta^{(i)}(y_{n};a) \right)_{i=m, \ldots, 2m-1}.
\]
In order to do this, an efficient way of computing $F(a)$ for a given $a = (a^{(0)}, \ldots  a^{(m-1)}) \in \Real^m$ is required. The following pseudocode illustrates the computation of $F(a)$ in the case $m = 2$, which is the value we mostly use for the numerical computations. It is straightforward to adapt the code for any value of $m \in \Natural$.

\bigskip

\texttt{Input} $a=(a^{(0)},\ a^{(1)})=(\Theta^{(0)}(y_{1}; a),\ \Theta^{(1)}(y_{1}; a))$.

\texttt{Set} $\Theta^{(2)}(y_{1};a) = 0, \ \Theta^{(3)}(y_{0};a) = 0.$

\texttt{For} $k= 1,\ldots,n - 1$, \texttt{set:}
\begin{align*}
\quad \quad \Theta^{(3)}(y_k; a) &= \Theta^{(3)}(y_{k-1}; a) + (n \lambda)^{-1} \Big(1 - r_{k}^{2}e^{2 \Theta^{(0)} (y_{k}; a)} \Big),\\
\quad \Theta^{(2)}(y_{k+1};a) &= \Theta^{(2)}(y_{k};a)+\Theta^{(3)}(y_{k};a)(y_{k+1}-y_{k}),\\
\quad \Theta^{(1)}(y_{k+1};a) &=\Theta^{(1)}(y_{k};a)+ \Theta^{(2)}(y_{k};a)(y_{k+1}-y_{k}) + (1/2) \Theta^{(3)}(y_{k};a) (y_{k+1}-y_{k})^2,\\
\quad \Theta^{(0)}(y_{k+1};a) &=\Theta^{(0)}(y_{k};a)+ \Theta^{(1)}(y_{k};a)(y_{k+1}-y_{k}) + (1/2) \Theta^{(2)}(y_{k};a) (y_{k+1}-y_{k})^2 \\
\ &+ (1/6) \Theta^{(3)}(y_{k};a) (y_{k+1}-y_{k})^3.
\end{align*}

\texttt{Next} k

\texttt{Set} $\Theta^{(3)}(y_n; a) = \Theta^{(3)}(y_{n-1}; a) + (n\lambda)^{-1} \Big(1 - r_{n}^{2}e^{\Theta^{(0)} (y_{n}; a)} \Big)$

\texttt{Return} $F(a)=(\Theta^{(2)}(y_{n};a),\Theta^{(3)}(y_{n};a))$.

\bigskip

Note that the philosophy of our algorithm is a variant of the so-called ``shooting method'' \cite[page 177]{asher:petzold:1988}. When close to the root $a_{*}$, Newton's method can be utilized to get a better ``aim,'' which will result in fast convergence of this iterative scheme. (Newton's method has to be used with care --- if far away from the root, the method is not likely to work and the numerical scheme will fail to converge. For further comments on this problem, see subsection \ref{ss:RegardingConvergence}.) Newton's method requires computation of the partial derivatives $(\partial/\partial a^{(i)})F$,which can be easily computed along with $F(a)$. Indeed, note that $\varphi_{i}(\cdot;a):=(\partial/\partial a^{(i)})\Theta^{(j)}(\cdot;a)$ for $i=0, \ldots, m-1$ and $j=0, \ldots, 2 m-1$ satisfy
\[
\varphi_{i}^{(j)}(y_1 ;a) = \frac{\partial \Theta^{(j)} (y_1)}{\partial a^{(i)}} = \frac{\partial a^{(j)}}{\partial a^{(i)}} = \left \{ 
	\begin{tabular}{ll}
		1, & if $i = j$ \\
		0, & if $i \neq j$ \\
	\end{tabular} \right..
\]
Furthermore, upon differentiating \eqref{eq:IterativeEquation}, we obtain
\[
\varphi_{i}^{(2m-1)}(y;a) = \frac{\partial \Theta^{(2m-1)} (y)}{\partial a^{(i)}} = \frac{(-1)^{m+1}}{n\lambda}\sum_{j=1}^{k}r_{j}^{2}{e^{{2\Theta^{(0)}(y_{j};a)}}}\varphi^{(0)}_{i}(y_{j};a)
\]
for all $y \in [y_{k},y_{k+1})$, where $k=1, \ldots,n$.
Therefore, computing these derivatives can be performed for little extra cost in the same iterative procedure used to compute $\Theta(y_{j};a)$ for $j=1, \ldots, n$. Indeed, one has to simply differentiate the iterative equations with respect to $a$ and obtain iterative equations for the derivatives
As soon as the information about the $m$ partial first-order derivatives of $F$ at a particular $a\in\mathbb{R}^{m}$ is obtained, we can implement Newton's method for finding the root of the equation $F(a)=0$. Newton's method proceeds by successive approximation. Consider an initial guess $a_{1}=(a_{1}^{(i)})_{i=0,...,m-1}$ and inputs $\delta > 0$ and $\epsilon > 0$, where $\delta$ controls the step size in Newton's method and $\epsilon$ determines when Newton's method is terminated.\footnote{In most of the examples in this paper, we took $\epsilon=10^{-10}$ and $\delta=0.1$, the latter to prevent overstepping in Newton's method.} Then, for $k = 1, \ldots$, while $|F(a_{k})| \geq \epsilon$ (where $|\cdot|$ denotes the usual Euclidean norm), one sets $a_{k+1}=a_{k} - \delta\Phi^{-1}(a_{k}) F(a_{k})$. In the previous formula, $\Phi$ is the $m \times m$ matrix with entry $\varphi_i^{(j)} (y_n; a)$ in the $i^{\textrm{th}}$ row and the $j^{\textrm{th}}$ column for $i=0, \ldots, m-1$ and $j=0, \ldots, m-1$. (Note that we are numbering rows and columns from $0$ to $m-1$, in order to be consistent with our notation.)

\begin{rem}
There is accompanying software that implements the above algorithm, together with instructions regarding its use, available upon request from the authors.
\end{rem}

\subsection{Choosing the penalization parameters}
\label{ss:ChoosingThePenalizationParameters}

In the literature on nonparametric regression using the penalized likelihood method (also called the regularization method), the most typical choice for the order of differentiation to penalize is $m=2$ --- see, for example, \cite{cox:osullivan:1990} or \cite{green:silverman:1994}. The choice $m=2$ gives rise to estimators that are natural cubic splines and are visually very attractive. The use of $m=1$ results in estimated functions that are quite ``wiggly'' --- see for example, Figure \ref{fig:BasicConvergenceExample}. Use of $m \geq 3$ is computationally involved; for this reason, we refrain from such practice.

The choice of the penalization coefficient $\lambda$ is more subtle. One can either use cross-validation techniques --- see, for example, \cite{fan:yao:2005} or \cite{green:silverman:1994}. However, it is often that case that simple visual inspection of the graphs is sufficient.

Theoretical results from the theory of nonparametric regression using the penalized likelihood method provide a hint in understanding how the penalization coefficient $\lambda$ should decrease with increasing sample size $n$ when the degree of smoothness (as given by $m$) of the target function is fixed. More precisely, for fixed $m$, in it is conjectured that the choice $\lambda_n \sim n^{-2m / (2m + 1)}$ as $n \to \infty$ will result in convergence of the estimator to the true value of the order $n^{- m / (2m+1)}$. For theoretical background, see \cite{cox:osullivan:1990}. In the next section, we shall provide empirical results on rates of convergence for $m=1$ and $m=2$ that seem to support this conjecture.

\subsection{Practical remarks regarding convergence of our algorithm}
\label{ss:RegardingConvergence}

If one uses the proposed penalization $\lambda_n \sim n^{-2m / (2m + 1)}$ for a given $m \in \Natural$ as $n \to \infty$, the numerical scheme suggested in Subsection \ref{ss: computing} tends to be unstable for large sample sizes. In order to get a feeling for the reason, note that $n \lambda_n \sim n^{1 / (2m+1)}$ converges to infinity when $n \to \infty$, albeit slower than $n$. Since one has to repeat the iteration given by \eqref{eq:IterativeEquationBetter} $n$ times to obtain $F (a)$, it is reasonable to suspect (and it actually happens in practice) that for choices of $a \in \Real^m$ that are far away from $a_*$ the vector $\Phi^{-1}(a) F(a)$ (recall that $\Phi$ is the $m \times m$ matrix of first order partial derivatives of $F$) consists of entries with huge magnitude. This situation results in failure of convergence of the algorithm, since the updating step in Newton's method takes one father away from the sought-after root. In order to overcome this difficulty, it is often desirable to first consider a subset of the sample, effectively thinning the observations. Once the data set has been thinned enough in order for the method to converge, one can use the root derived from the thinned data set as initial seed for a denser subset of the data. Continuing this way, one finally computes the estimator of $\sigma$ using the whole sample. Note that we use this method in order to obtain the estimators in the following two sections.

\section{Simulation Study and Empirical Rates of Convergence}
\label{s:SimulationStudy}

In this section, we shall conduct an empirical investigation of the rate at which our estimators converge to the true volatility function.

As a simple benchmark case when $m=1$, we consider the case of a driftless Brownian motion with $\sigma \equiv 3$. We generate a path of the Brownian motion using \emph{exact} simulation, and produce data points with a step size of $\Delta t =2^{-17}$ over the unit interval. Continuing, we remove every other realization of the sample path to arrive at a ``reduction'' of the original sample path with step sizes of $\Delta t=2^{-16}$. We then remove every other observation three more times to create yet another reduction with step size $\Delta t=2^{-13}$. The resulting estimates $\sigma_{*}$ of the true volatility function $\sigma \equiv 3$ are shown in Figure \ref{fig:BasicConvergenceExample} using $\lambda = 20 (\Delta t)^{2/3}$. (This choice is consistent with the discussion in subsection \ref{ss:ChoosingThePenalizationParameters}.) Notice that convergence to the constant volatility function $\sigma\equiv3$ seems to be quite fast.

\begin{figure}
 \centering
   \includegraphics[width=0.50\textwidth]{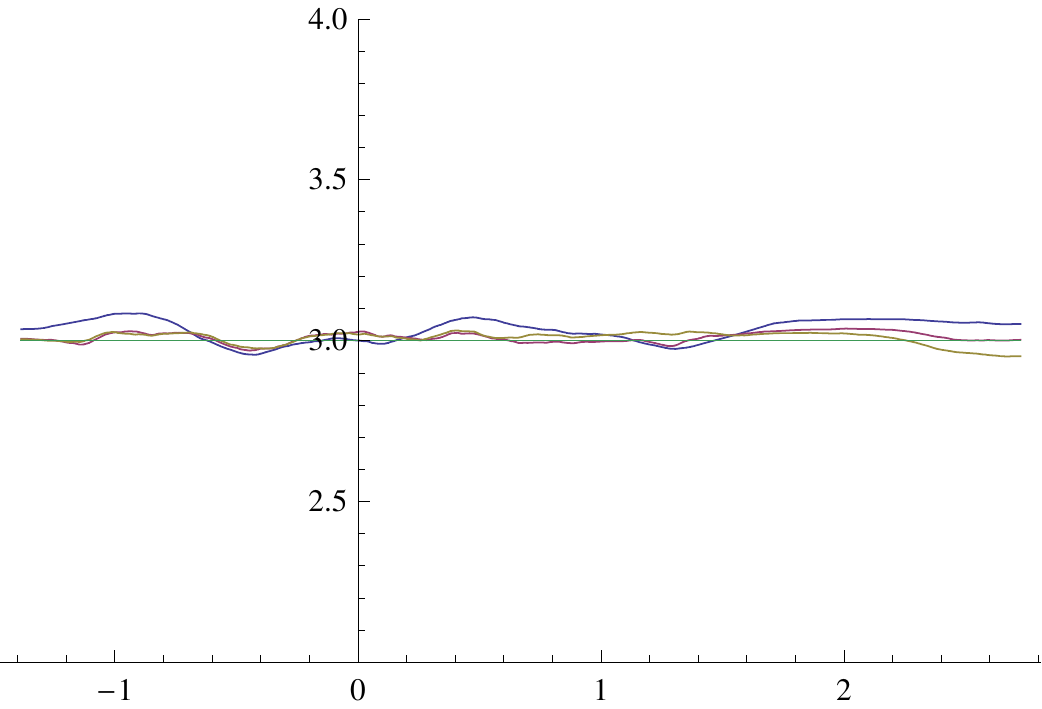}
 \setcaptionwidth{0.85\textwidth}
 \caption{We generate a sample path with $\Delta t=1/2^{17}$, no drift, and $\sigma \equiv 3$ and we also create two reductions of this sample path with $\Delta t=1/2^{16}$ and $\Delta t=1/2^{13}$. Along with the true $\sigma$, the resulting $\sigma_{*}$ are shown when $\lambda = 20 (\Delta t)^{2/3}$ and $m=1$. 
}
\label{fig:BasicConvergenceExample}
\end{figure}

To study the empirical rates of convergence more precisely, and to illustrate that our method works well even for diffusions with non-zero drift, we present a more involved example. Consider the diffusion $Y$ with $Y_0 = 1/2$ and dynamics
\begin{equation} \label{eq: cool diff}
\ud Y_t = - Y_t^2 (1 - Y_t) \ud t + Y_t (1 - Y_t) \ud W_t, \quad t \in [0, 1].
\end{equation}
The diffusion $Y$ is $(0,1)$-valued; in fact, a straightforward use of It\^o's formula shows that
\begin{equation} \label{eq: cool diff closed form}
Y_t = \frac{\exp(W_t - t/2)}{1 + \exp(W_t - t/2)}, \quad t \in [0, 1].
\end{equation}
Note that $Y$ does \emph{not} have zero drift and that $\sigma(y) = y(1 - y)$. Furthermore, equation \eqref{eq: cool diff closed form} gives a way to simulate $Y$ without any discretization error, since one only needs to simulate $W$, which can be done exactly. We simulate a sample path with a step size $2^{-25}$ and $T=1$, and also create  reduced versions of this sample path with step sizes $\Delta t$ of $2^{-10}, 2^{-11}, \ldots, 2^{-25}$. For the case $m=1$, we use $\lambda = 30(\Delta t)^{ 2 / 3}$, while for the case $m=2$ we use $\lambda = 20(\Delta t)^{ 4/5}$.
\begin{figure}
 \centering
 \subfloat[]{\label{fig:rateofconvergence-samplepath02-onedimenaional}
   \includegraphics[width=0.40\textwidth]{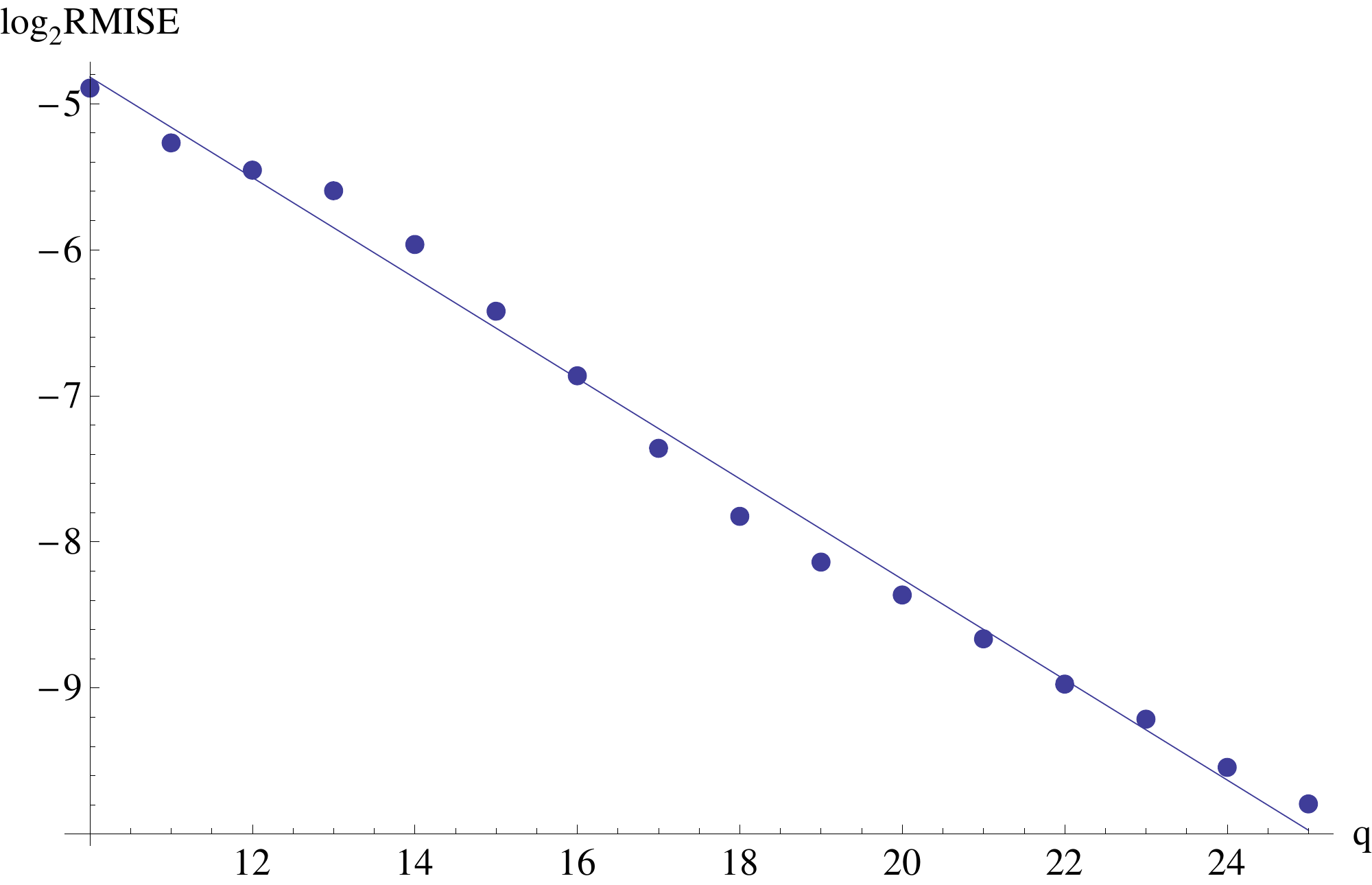}}
   \quad%
 \subfloat[]{\label{fig:rateofconvergence-samplepath02-twodimensional}
   \includegraphics[width=0.40\textwidth]{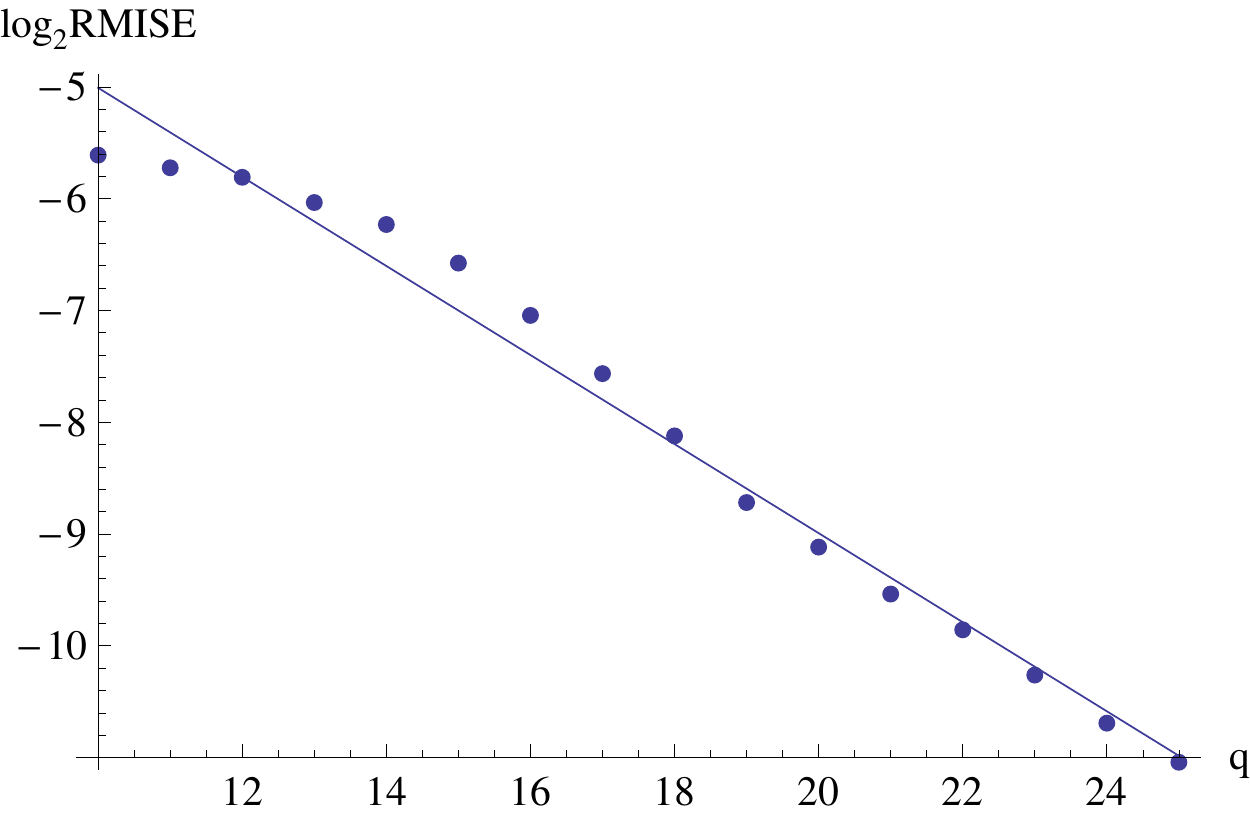}}
   \quad%
 \setcaptionwidth{0.85\textwidth}
 \caption{We generate a sample path with $\Delta t=1/2^{25}$ from the diffusion whose dynamics are given by equation \eqref{eq: cool diff}. Then, we create $15$ reductions of this sample path with $\Delta t=1/2^{10},...,1/2^{25}$. (A) In the case $m = 1$, we plot $\log_{2}{\textrm{RMSE}(\sigma_{*}^{(1,q)},\sigma)}$ against $q$ for $q=10,...,25$. The least-squares line is $-1.383-0.343 q$, which is consistent with results from kernel-based estimation schemes which suggest that the slope should be $- 1/3 \approx - 0.333$. (B) In the case $m=2$, we again plot $\log_{2}{\textrm{RMSE}(\sigma_{*}^{(2,q)},\sigma)}$ against $q$ for $q=10,...,25$. The least-squares line is $-1.024-0.398 q$, which is consistent with our conjecture which suggests that the slope should be $- 2/5 = - 0.4$.}
 \label{fig:RateOfConvergence}
\end{figure}
For both $m=1$ and $m=2$ and $q=10, \ldots, 25$, we produce penalized maximum quasi-likelihood estimators $\sigma_{*}^{(m,q)}$. We then compute the root mean-integrated squared error
\[
\textrm{RMISE}(\sigma_{*}^{(m,q)},\sigma) = \sqrt{ \frac{T}{n} \sum_{j = 1}^{n} \left( \sigma_{*}^{(m,q)}(y_j)-\sigma(y_j) \right)^{2} }
\]
where $n = 2^{q} T$ is the sample size, as a proxy to the quantity
\[
\sqrt{\int_{0}^{T}(\sigma_{*}^{(m,q)}(Y_t)-\sigma(Y_t))^{2} d t} = \sqrt{\int_{m_{T}}^{M_{T}} \left( \frac{\sigma_{*}^{(m,q)}(y)}{\sigma(y)} - 1 \right)^{2} L_T^Y (y) d y},
\]
where $L_T^Y (y)$ is the semimartingale local time of $Y$ at level $y \in \Real$ accumulated up to time $T$. We then plot $\log_{2}(\textrm{RMISE}(\sigma_{*}^{(m,q)},\sigma))$ against $q$ and execute least-squares linear fits to determine the rate of convergence as we ``fill in'' the sample path --- see Figure \ref{fig:RateOfConvergence}. The regression line in the case $m=1$ was $-1.383-0.343 q$.  In the case $m=2$, the regression line is $-1.024-0.398 q$. In both cases, the slope of the line is roughly consistent with convergence results for kernel-based estimation schemes, which suggests that the slope should be $-m/(2m+1)$, giving $-1/3$ for $m=1$ and $-2/5$ for $m=2$. We, therefore, conjecture (as already mentioned) that the rate of convergence is of order $n^{- m / (2m+1)}$ when one chooses $m$ as penalization differentiation order and uses $\lambda_{n} \sim n^{ - 2m / (2m+1)}$ as $n \to \infty$. Note that this conjectured rate is optimal, as demonstrated in \cite{hoffmann:1999}.

\section{Application to Exchange Rates and Interest Rates}
\label{s:ApplicationToRealData}

\subsection{Exchange Rates.} The three exchange rates were taken from the Federal Reserve Economic Data (FRED) database, which is maintained by the St. Louis branch of the Federal Reserve Bank. (The FRED database can be accessed at \url{http://research.stlouisfed.org/fred2}. Registering for a username and password is required as of May 31, 2010.) For purposes of comparability, we study all three exchange rates from January 4, 1999 through May 21, 2010 (the Euro debuted on at the beginning of 1999). The FRED database contains noon buying rates in New York City for cable transfers payable in foreign currencies. It includes data only for days on which financial markets are open.


Though exchange rate data at weekly and monthly reporting frequencies is also available through the FRED database, our nonparametric estimation technique performs best when the frequency of the data is relatively high. However, we did not attempt to use intra-day exchange rate data, even though such data is increasingly available for free on the Internet. (See, for example, the weblink \url{http://www.forexrate.co.uk/forexhistoricaldata.php}.) Though bid-ask spreads on exchange rates are generally quite low, much of the apparent volatility observed in financial asset prices observed with sufficiently high frequency (for example, one-minute time intervals) is due to buyers buying at the ask price and sellers selling at the bid price. This ``toggling'' between bid and ask prices during very high frequency trading biases volatility estimation upwards, rendering rules like the ``square root of time multiplied by the volatility'' ineffectual --- see \cite{jorion:2009}.

We first treat the data by computing the $R_{t_{i}}$ mentioned in equation (\ref{eq:PointwiseVariances}). In particular, we are careful to compute $\sqrt{t_{i}-t_{i-1}}$ precisely and in calendar time, rather than by assuming that a calendar year has approximately 250 trading days and then dividing $Y_{t_{i}}-Y_{t_{i-1}}$ by $\sqrt{1/250}$. This adjustment is small but potentially important, particularly when using high frequency data. For example, the standard deviation of the returns on the USD/EUR exchange rate over weekdays (or holidays) is $0.6457733\%$. Over weekends, the standard deviation is $0.675802\%$. With 624 weekend/holiday trading days and 2866 weekday returns, the variance ratio test yields a test statistic of 1.0953 and a rejection of the null hypothesis of equality of variances at the 10\% (though not the 5\%) level of significance. In general, we will appropriately adjust for the length of the underlying time intervals, since there appears to be a slight but meaningful tendency for information to be released over the weekend that causes Friday-to-Monday returns to have somewhat higher dispersion than regular weekday returns.

We also pre-process the data in one additional way to make it suitable for the numerical scheme articulated in Section \ref{s:VolatilityEstimationForDiffusionProcess}. In this section, we required that sorted values of the asset prices (denoted $y_{j}$) be distinct from one another. In almost all diffusion models, the theoretical probability that two values $Y_{t_{i}}$ and $Y_{t_{j}}$ are equal is zero when $i \neq j$. In practical situations in financial markets, in which asset prices are only recorded to a finite number of decimal places, ties are possible and sorting the raw asset price data becomes an ill-defined task. To manage this problem, we tentatively compute the values of $R_{t_{i}}^{2}$, take their mean, and add to the raw asset price data a Gaussian random number with mean zero and variance equal to a very small constant times the mean of the $R_{t_{i}}^{2}$. We then recompute $|R_{t_{i}}|$. By slightly perturbing the raw data, all ties are randomly broken and, in practice, the movement in the time series is very negligible. (One could also perform the perturbation only to the tied data --- in practice, the two approaches lead to almost identical output.)

\begin{figure}
 \centering
 \subfloat[]{\label{fig:USDEURExchangeRateVolatility}
   \includegraphics[width=0.45\textwidth]{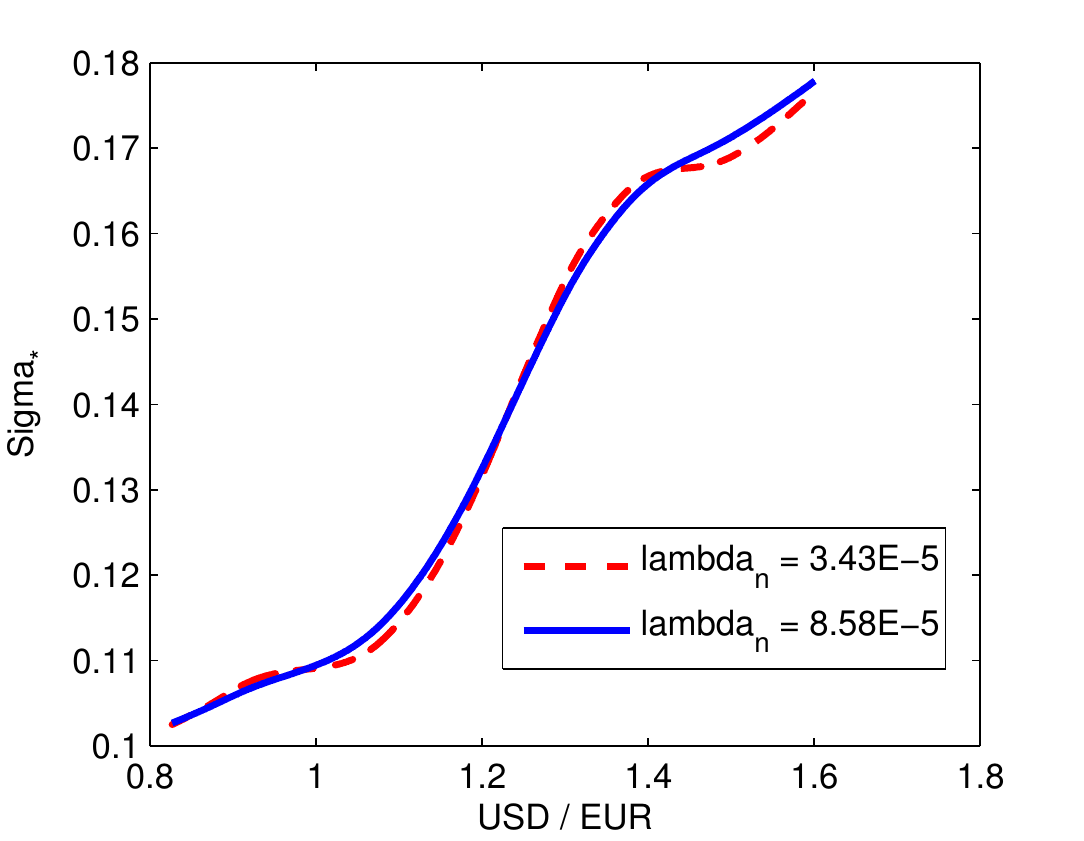}}
   \quad%
 \subfloat[]{\label{fig:USDGBPExchangeRateVolatility}
   \includegraphics[width=0.45\textwidth]{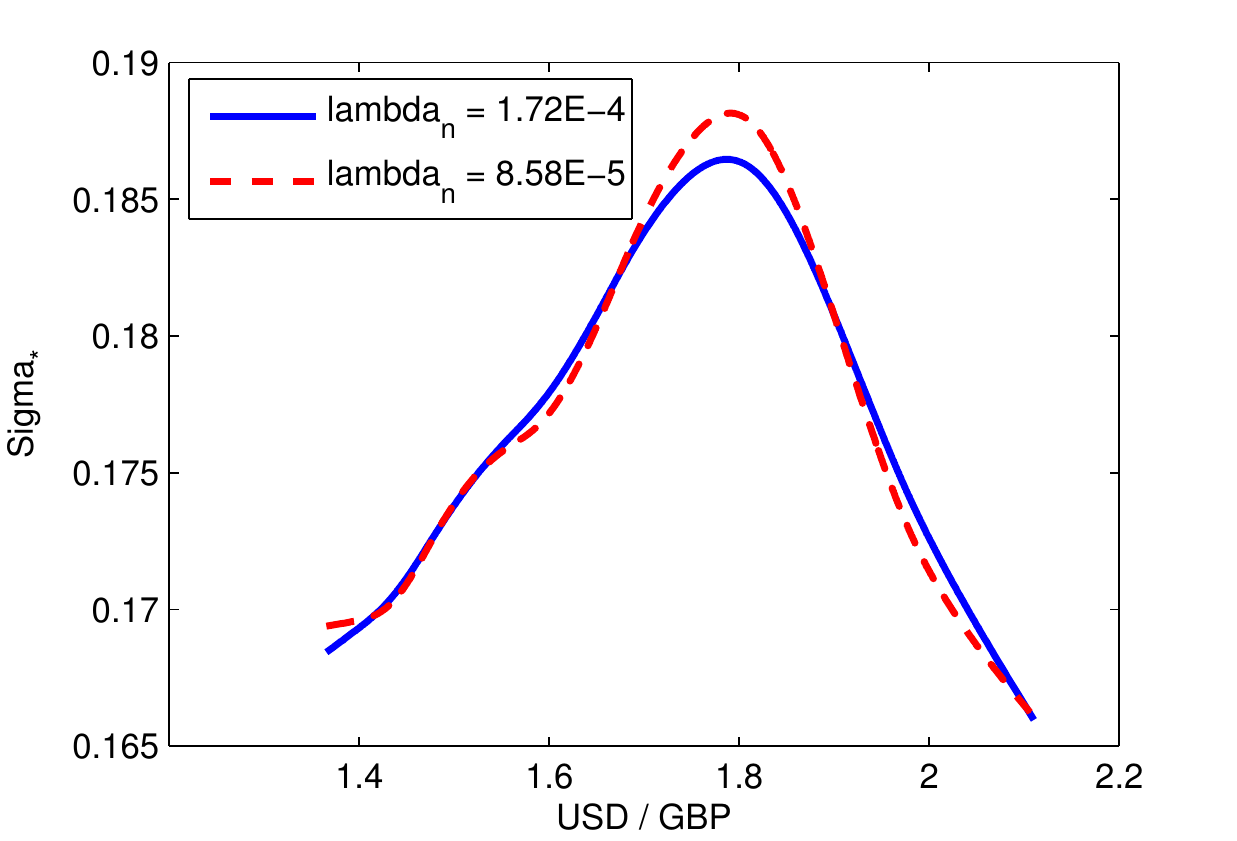}}
   \quad%
 \subfloat[]{\label{fig:JPYUSDExchangeRateVolatility}
   \includegraphics[width=0.45\textwidth]{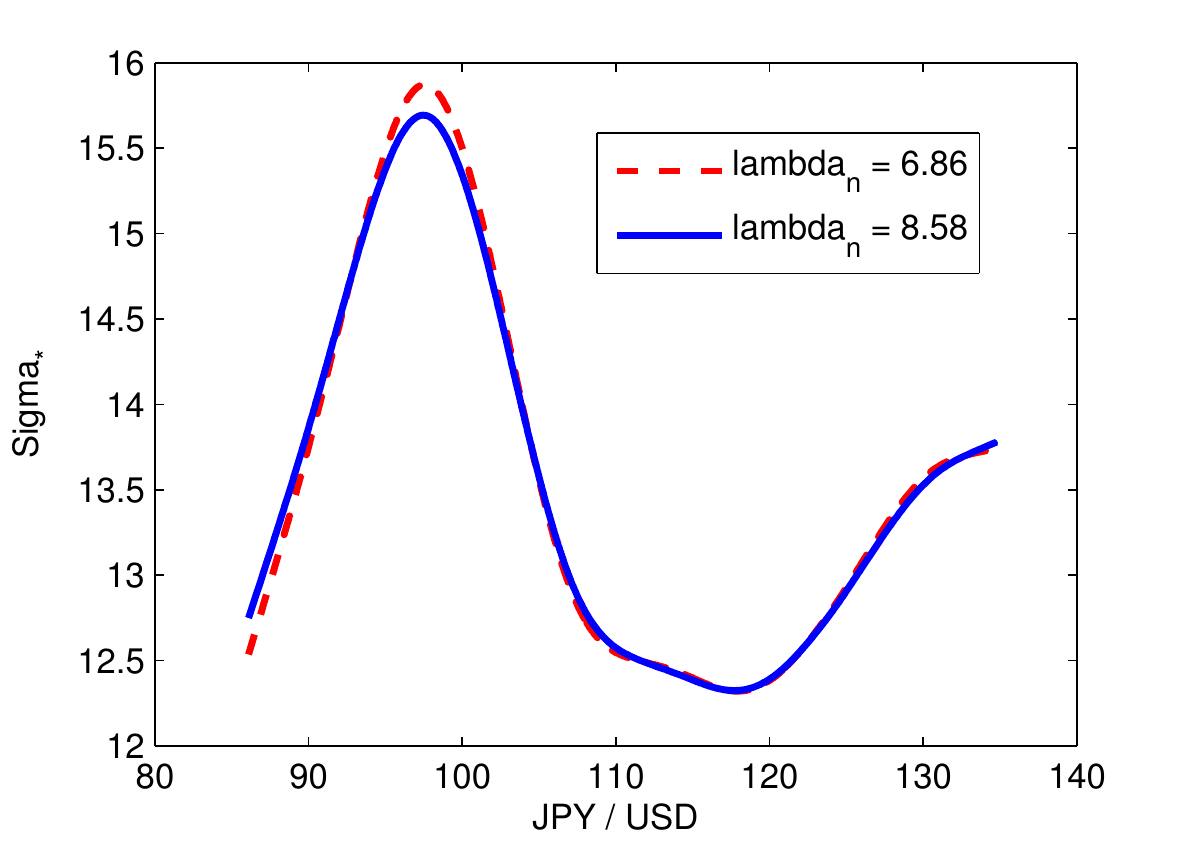}}
   \quad%
 \subfloat[]{\label{fig:EURUSDExchangeRateVolatility}
   \includegraphics[width=0.45\textwidth]{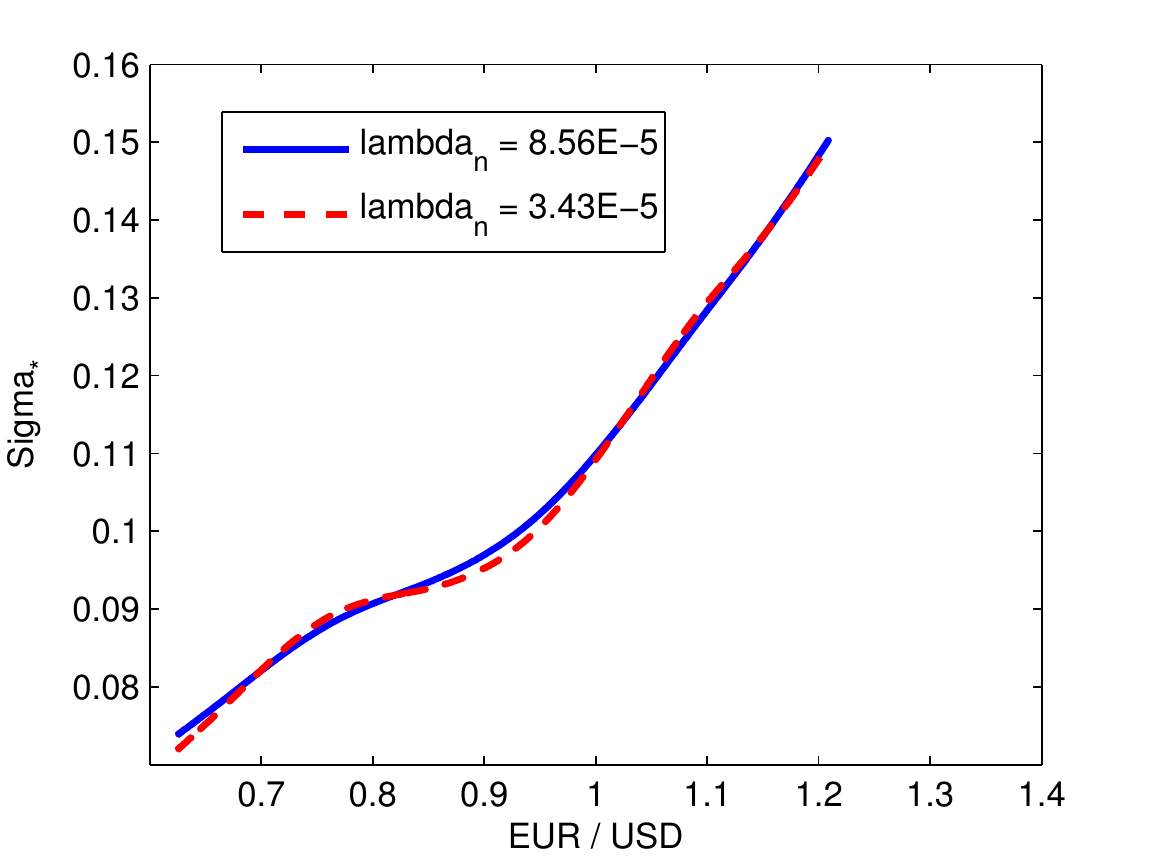}}
   \quad%
 \setcaptionwidth{0.85\textwidth}
 \caption{Under the assumption that the USD/EUR, USD/GBP, JPY/USD, and EUR/USD exchange rates for the period from January 4, 1999 to May 21, 2010 are governed by diffusion processes, these figures show estimates $\sigma_{*}$ of the diffusion function $\sigma$.}
 \label{fig:VariousExchangeRateVolatilities}
\end{figure}

In Figure \ref{fig:USDEURExchangeRateVolatility}, we plot the estimates of the USD/EUR exchange rate volatility for two different values of $\lambda$ for the period from January 4, 1999 to May 21, 2010. Consistent with the discussion in Section \ref{s:SimulationStudy}, we choose $\lambda$ to be $0.02n^{-4/5}$ and $0.05n^{-4/5}$, where $n$ is equal to the number of returns in our time series --- 2866 days, to be precise. 
The two estimates of the volatility in Figure \ref{fig:USDEURExchangeRateVolatility} seem to indicate that using a linear function for $\sigma$ in equation (\ref{eq:Diffusion}) is not appropriate. Instead, the volatility of the USD/EUR exchange rate seems to increase at a fairly rapid pace around values of $1.30$ USD/EUR.

In Figure \ref{fig:USDGBPExchangeRateVolatility}, we similarly show the estimates $\sigma_{*}$ for the USD/GDP exchange rate. The USD/GDP exchange rate volatility varies between 16.5\% and 19\% over $[y_1,y_n]$ from equation \ref{eq:MinimumAndMaximum}. Note that the while the estimated diffusion function for the USD/EUR exchange rate appears to be monotonically increasing, the estimated USD/GBP exchange rate is not. Generally, diffusion models of financial phenomena assume that the diffusion function is either constant (like the Vasicek model) or that the diffusion function is a constant multiplied by some increasing function (geometric Brownian motion or the Cox-Ingersoll-Ross process). We note that it would be interesting, and potentially the focus of future work, to use our estimated diffusion functions in hypothesis testing --- for example, in this setting one could test whether the assumption that the USD/GDP exchange rate volatility is constant can be rejected at a  high level of statistical significance.

Finally, in Figure \ref{fig:JPYUSDExchangeRateVolatility} and \ref{fig:EURUSDExchangeRateVolatility}, we show the estimates $\sigma_{*}$ for the JPY/USD and EUR/USD exchange rates, respectively. Note that the estimated diffusion function for the JPY/USD exchange rate is also not monotonic and that it peaks when the exchange rate is between 95-100 Japanese yen per U.S. dollar. The estimated diffusion function associated to the EUR/USD exchange rate looks fairly similar to the function in Figure \ref{fig:USDEURExchangeRateVolatility}.

\subsection{LIBOR Rates, Treasury Bill Yields, and Treasury Bond Yields.} We obtain constant-maturity 1-month U.S. Treasury bill yields, constant-maturity 3-month U.S. Treasury bill yields, and constant-maturity 30-year Treasury bond yields from the Federal Reserve Economic Data (FRED) database as well. (The time series designations in the FRED database for the constant-maturity 1-month U.S. Treasury bill yields, the constant-maturity 3-month U.S. Treasury bill yields, and 30-year U.S. Treasury bond yields are \texttt{DGS1MO}, \texttt{DGS3MO}, and \texttt{GS30}, respectively.) Again, the yields are recorded each day at 12:00 noon Eastern Standard Time. The FRED database contains yields recorded at the end of every week and the end of every month, but we chose to work with yields recorded at the end of each day on which financial markets were open. Additionally, we computed the $R_{t_{i-1}}$ in equation (\ref{eq:PointwiseVariances}) correctly, fully accounting for long weekends and market holidays. We slightly randomly perturbed the raw data to break ties, as discussed in the previous section. Unlike the Treasury bill and bond yield data, we obtain overnight London Interbank Offered Rates (LIBOR) from Thompson Reuters Datastream. (For more information about Datastream, see \url{http://thomsonreuters.com/products}.) The LIBOR rate is the rate that large banks use to borrow and lend from one another on the overnight market.

\begin{figure}
 \centering
 \subfloat[]{\label{fig:LIBORVolatility}
   \includegraphics[width=0.45\textwidth]{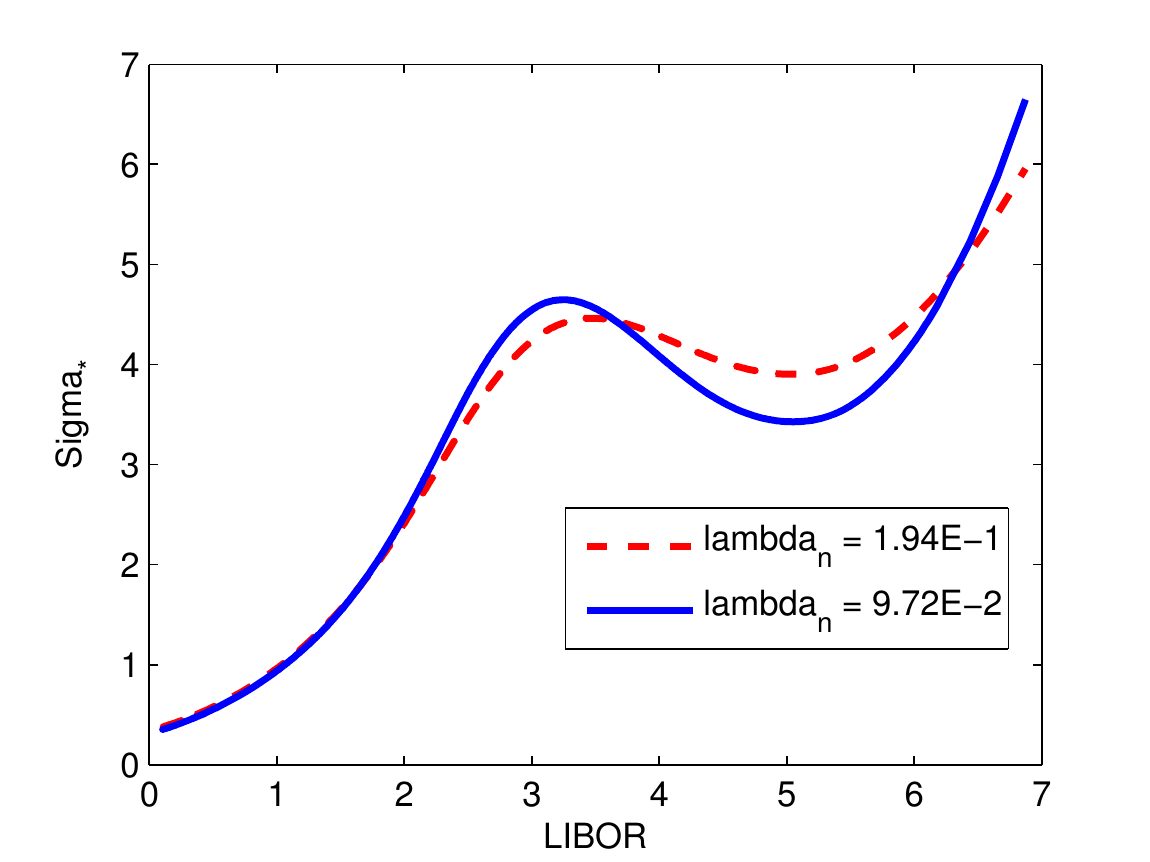}}
   \quad%
 \subfloat[]{\label{fig:OneMonthTreasuryBillYieldVolatility}
   \includegraphics[width=0.45\textwidth]{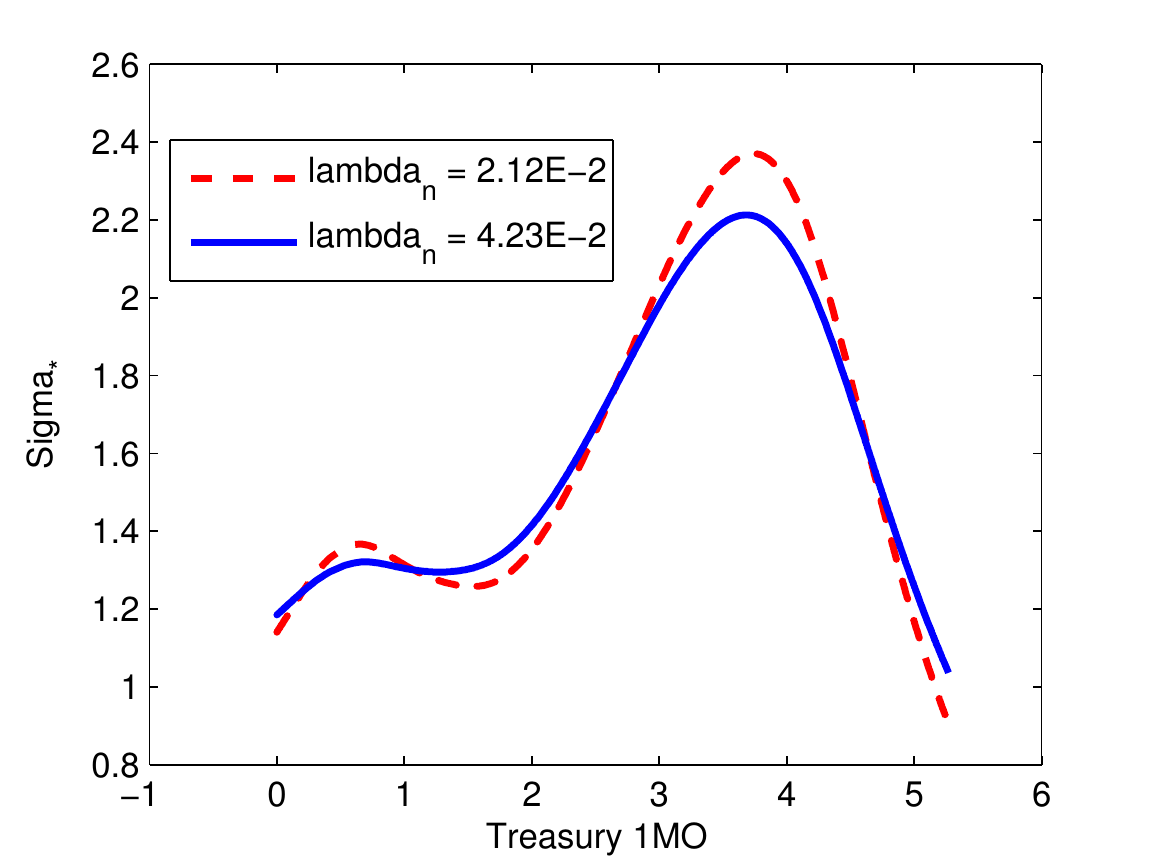}}
   \quad%
 \subfloat[]{\label{fig:ThreeMonthTreasuryBillYieldVolatility}
   \includegraphics[width=0.45\textwidth]{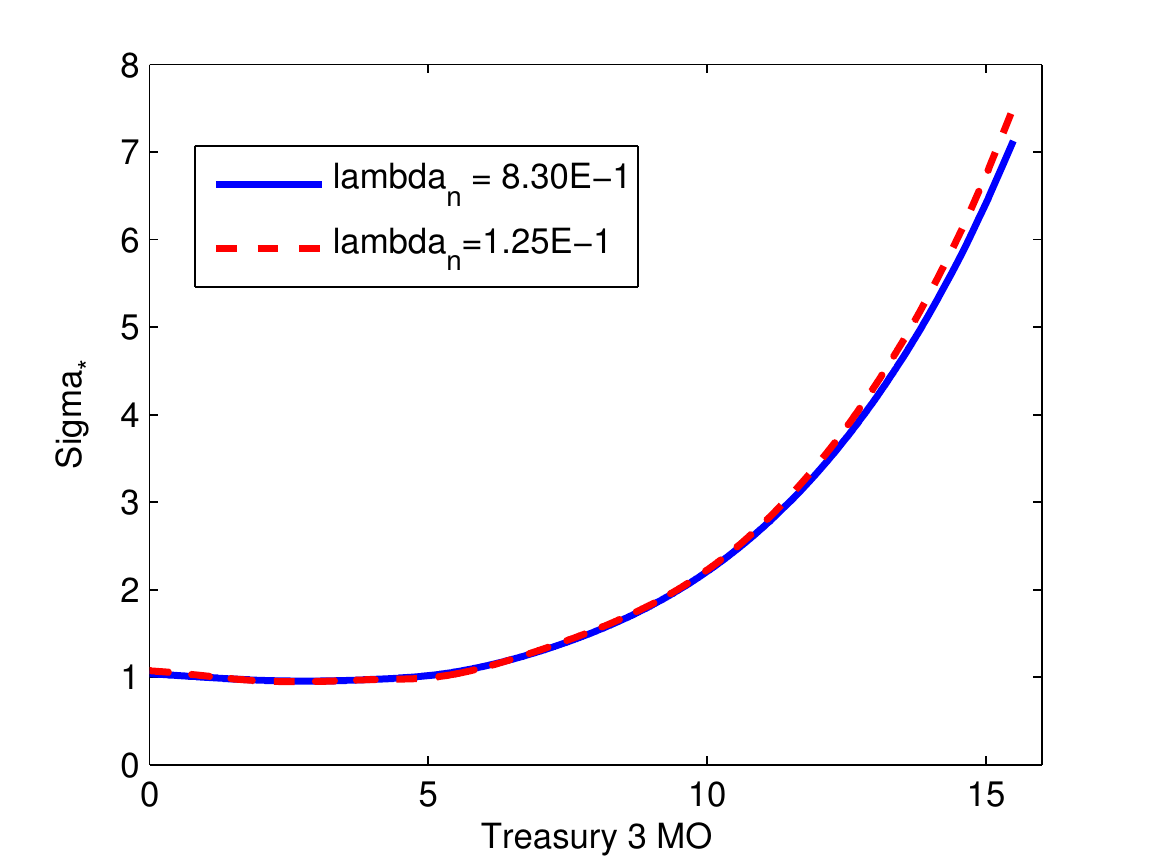}}
   \quad%
 \subfloat[]{\label{fig:ThirtyYearTreasuryBondYieldVolatility}
   \includegraphics[width=0.45\textwidth]{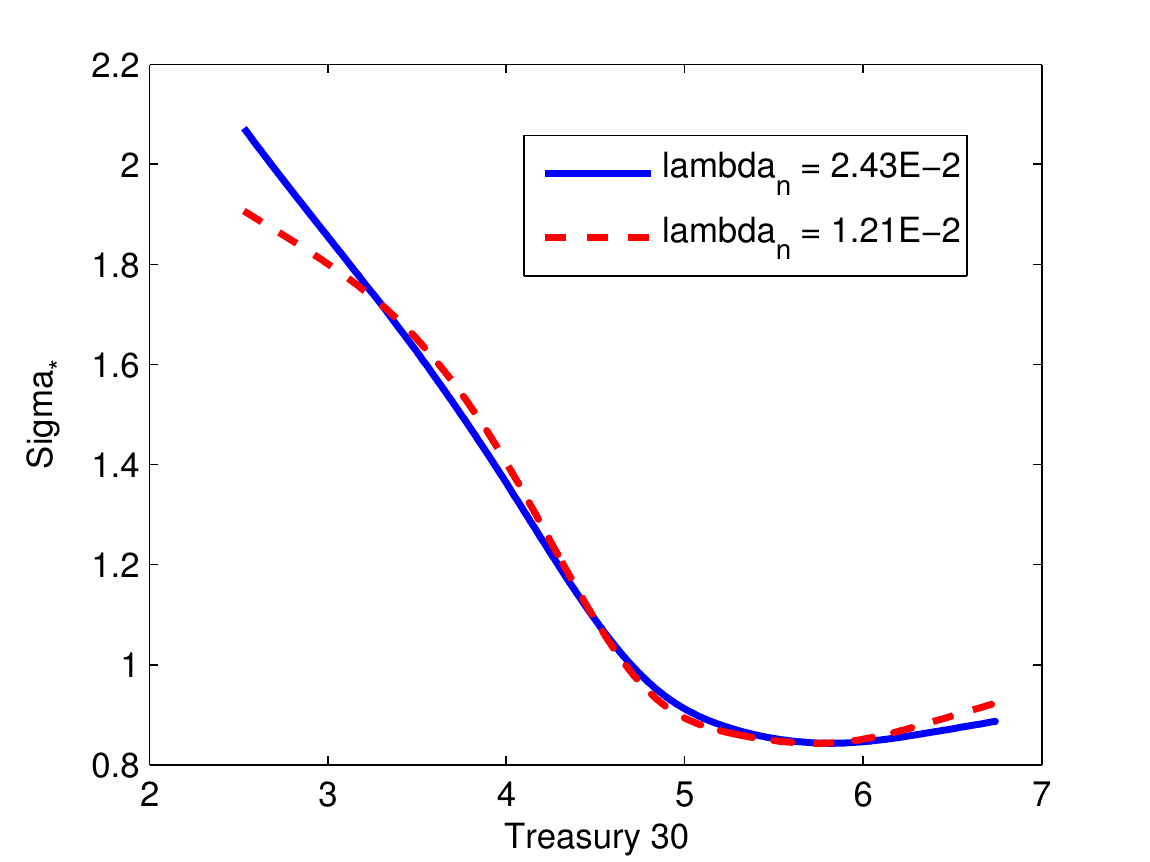}}
   \quad%
 \setcaptionwidth{0.85\textwidth}
 \caption{Under the assumption that overnight LIBOR, 1-month U.S. Treasury bill yields, 3-month U.S. Treasury bill yields, and 30-year U.S. Treasury bond yields are governed by diffusion processes, these figures show estimates $\sigma_{*}$ of the diffusion function $\sigma$. The periods of time for which the estimates $\sigma_{*}$ were derived are, respectively, January 2, 1999 through May 26, 2010 for LIBOR; July 31, 2001 through May 31, 2010 for 1-month U.S. Treasury bill yields; January 4, 1982 through Mary 21, 2010 for 3-month U.S. Treasury bill yields; and January 4, 1999 through May 21, 2010 for 30-year U.S. Treasury bond yields.}
 \label{fig:VariousInterestRateVolatilities}
\end{figure}


The periods of time for which the estimates $\sigma_{*}$ were derived were: January 2, 1999 through May 26, 2010 for LIBOR, July 31, 2001 through May 31, 2010 for 1-month U.S. Treasury bill yields, January 4, 1982 through Mary 21, 2010 for 3-month U.S. Treasury bill yields, and January 4, 1999 through May 21, 2010 for 30-year U.S. Treasury bond yields.

In Figure \ref{fig:LIBORVolatility}, we plot the estimates of London Interbank Offered Rate (LIBOR) volatility for two different values of $\lambda$. Our procedure generates an estimate of the volatility function, $\sigma_{*}$, that is clearly neither linear nor a constant multiple of the square root function (as in the case of the Cox-Ingersoll-Ross process). Indeed, the volatility function is not even monotonic. Though similarly non-monotonic, the estimate $\sigma_{*}$ for one-month U.S. Treasury bill yields is low when yields are both relatively low and relatively high. The volatility of constant-maturity one-month U.S. Treasury bill yields seems to be highest when yields are around $3.75\%$. The estimates of $\sigma$ for 3-month U.S. Treasury bill yields and 30-year Treasury bond yields are remarkably different. The volatility function associated with 3-month U.S. Treasury bill yields appears to be a monotonic increasing function of the underlying yield. In contrast, the volatility function associated with 30-year U.S. Treasury bond yields is monotonically decreasing function of the underlying yield yield. No function appears to adhere closely to any of the standard choices of $\sigma$ in the financial literature.

\bibliographystyle{siam}
\bibliography{mpqle_diff}

\end{document}